\documentclass[aps,prl,singlecolumn,nofootinbib]{revtex4-2}

\def\mnras{{MNRAS}}             
\def\prb{{Phys.~Rev.~B}}        
\def\prd{{Phys.~Rev.~D}}        
\def\pre{{Phys.~Rev.~E}}        
\def\prl{{Phys.~Rev.~Lett.}}    
\usepackage{amsmath,amstext}
\usepackage[T1]{fontenc}
\usepackage{amssymb}
\usepackage{graphicx}
\usepackage{ae,aecompl}

\usepackage{hyperref}
\usepackage{amsmath}
\usepackage{amssymb}
\usepackage{mathtools}
\usepackage{bm}
\usepackage{cleveref}
\usepackage{tensor}
\usepackage{braket}
\usepackage{enumitem}
\usepackage{mhchem}
\usepackage{cancel}
\usepackage{amsthm}
\usepackage{upgreek}
\usepackage{multirow}

\usepackage{color}

\newcommand{\spc}{\quad \quad \quad}

\def\be{\begin{equation}}
\def\ee{\end{equation}}
\def\beq{\begin{eqnarray}}
\def\eeq{\end{eqnarray}}

\theoremstyle{definition}

\theoremstyle{theorem}
\newtheorem{theorem}{Theorem}

\usepackage{soul}
\setstcolor{blue}

\begin{document}
\title{Universality Classes of Relativistic Fluid Dynamics: Foundations}
\author{L.~Gavassino$^1$, M.~Disconzi$^{1}$, \& J.~Noronha$^2$}
\affiliation{
$^1$Department of Mathematics, Vanderbilt University, Nashville, TN, USA
\\
$^2$Illinois Center for Advanced Studies of the Universe \& Department of Physics,
University of Illinois at Urbana-Champaign, Urbana, IL 61801-3003, USA
}

\begin{abstract}
A general organizing principle is proposed that can be used to derive the equations of motion describing the near-equilibrium dynamics of causal and thermodynamically stable relativistic systems. The latter are found to display some new type of universal behavior near equilibrium that allows them to be grouped into universality classes defined by their degrees of freedom, information content, and conservation laws. The universality classes expose a number of surprising equivalences between different theories, shedding new light on the near-equilibrium behavior of relativistic systems. 
\end{abstract} 
%
%
%
%
\maketitle
\section{Introduction}
The construction of a relativistic hydrodynamic theory usually relies on two choices. First, one must choose the degrees of freedom that describe the hydrodynamic state. Then, a guiding principle is used to derive the dynamical equations for those fields. Such guiding principle may be, e.g., a thermodynamic principle \cite{Israel_Stewart_1979,Muller_book,Otting1_1998,Salazar2020,GavassinoFronntiers2021}, a variational principle \cite{noto_rel,Dubosky2012_effective,TorrieriIS2016,GavassinoRadiazione}, kinetic theory \cite{Denicol2012Boltzmann}, holography \cite{Bhattacharyya:2007vjd,Baier2008}, or effective theory arguments \cite{Bemfica2017TheFirst,Kovtun2019,Bemfica:2019knx,Hoult:2020eho,Bemfica-Disconzi-Graber-2021,BemficaDNDefinitivo2020}. The level of freedom involved in these two choices is immense \cite{Geroch_Lindblom_1991_causal}, leading to a plethora of alternative theories  
, which are constantly being added to the ``relativistic hydrodynamics landscape'' \cite{Israel_Stewart_1979,Baier2008,Denicol2012Boltzmann,Florkowski:2010cf,Martinez:2010sc,Jaiswal:2013vta,Bemfica2017TheFirst,Kovtun2019,Bemfica:2019knx,Hoult:2020eho,BemficaDNDefinitivo2020,Florkowski_2018,
RauWasserman2020,Kiamari_2021,SperanzaChiral2021,
NoronhaGeneralFrame2021,Perna_2021,
GavassinoKhalatnikov2022,KeYin2022,
Singh2022,Most_2022,
HellerSingulant2022,Brito_2022,Wagner:2022ayd,SalazarZannias2022}. 
Nevertheless, it is known that some theories can become equivalent in certain regimes. For instance, second-order theories for viscous hydrodynamics arising from very different choices of both fields and dynamical equations  \cite{Hishcock1983,OlsonLifsh1990,GerochLindblom1990,
carter1991,Baier2008,Denicol2012Boltzmann,
Stricker2019} are mathematically equivalent near equilibrium \cite{Liu1986,PriouCOMPAR1991,GavassinoNonHydro2022,GavassinoGENERIC2022}. Equivalency here means that, although those theories can be different in the nonlinear regime, they become indistinguishable when linearized around homogeneous equilibrium states. Demonstrating such equivalency sometimes requires making a complicated change of variables so that, almost ``magically''  (see, e.g.,  \cite{GavassinoGENERIC2022}), one linearized theory is transformed into the other.
The fact that such transformations are possible near equilibrium cannot be some fortuitous coincidence. Rather, this should follow as a consequence of the underlying properties of equilibrium states in relativity.

Here, we present a general organizing principle that can be used to derive the equations of motion (EOM) describing the near-equilibrium dynamics of any causal and thermodynamically stable relativistic system. The method is based on a new result, rigorously proven here, which establishes that the EOM describing the system's linearized disturbances can always be obtained from a 4-vector field $E^\mu$ known as the ``information current'' \cite{GavassinoCausality2021}, which is a quadratic function of the perturbations (e.g., temperature variations). Assuming an isotropic and homogeneous equilibrium state, we show how the information current can be systematically constructed for arbitrary theories in terms of the corresponding linear-order perturbation fields, grouped according to their transformation properties under the $SO(3)$ rotation group. Entropy production follows from $\partial_\mu E^\mu \leq 0$,  which describes the fact that our initial information about the microstate of the system \cite{Jaynes1965} is erased as all microstates evolve towards the equilibrium state. Furthermore, we prove that, no matter how complicated a hydrodynamic theory is, one can always rearrange its equations of motion so that they resemble Israel-Stewart theory \cite{Israel_Stewart_1979} near equilibrium. This is used to reveal that causal and thermodynamically stable theories of relativistic fluid dynamics possess a new type of universal behavior near equilibrium that allows them to be grouped into universality classes. Each class describes a physically different behavior, defined solely by the degrees of freedom, the corresponding information current, and the conservation laws. The existence of such universality classes unveils a number of startling equivalences between seemingly different sets of equations of motion. 
Finally, as an application, we use our approach to show that (in the linear regime) an isotropic solid may be viewed as a fluid with an additional conserved charge, which transforms as a symmetric $(0,2)$-tensor under $SO(3)$.
Many more specific examples can be worked out for different universality classes (see accompanying paper). \emph{Notation}: We use $\hbar=k_B = c=1$, a $(-,+,+,+)$ Minkowski metric, Greek indices run from 0 to 3,  lowercase Latin indices from 1 to 3. Uppercase Latin indices are multi-index labels. $M_{(AB)}$ and $M_{[AB]}$ denote symmetric and antisymmetric parts of $M_{AB}$.

\section{Hyperbolicity from thermodynamics}
To understand how the universality classes come about, we need to derive some new results concerning the properties of the EOM of relativistic systems near equilibrium. We first recap some general facts about relativistic thermodynamics. All states of thermodynamic equilibrium maximize some thermodynamic potential $\Phi$ \cite{Stuekelberg1962,Israel_2009_inbook,Grmela1997,
huang_book,landau5,GavassinoTermometri}. For equations linearized about equilibrium, the quantity $E:=\Phi_\text{eq}-\Phi \geq 0$ plays the role of a non-increasing Lyapunov functional \cite{GavassinoGibbs2021}, which can be expressed as a volume integral $E(\Sigma) = \int_\Sigma E^\mu \, d\Sigma_\mu$, 
where $\Sigma$ is a Cauchy surface, $d\Sigma_\mu$ is its normal surface element (with orientation $d\Sigma_0>0$ \cite{MTW_book}), and $E^\mu$ is the information current \cite{GavassinoCausality2021} (see the Supplemental Material for a brief review). For the state of thermodynamic equilibrium to be stable against perturbations in all reference frames \cite{GavassinoSuperluminal2021,GavassinoBounds2023}, $E^\mu$ must be future-directed and timelike for any non-vanishing perturbation \cite{GavassinoCausality2021}. The second law of thermodynamics also requires 
\begin{equation}\label{Emu}
\partial_\mu E^\mu = -\sigma \leq 0 \, ,
\end{equation}
where $\sigma$ is the entropy production rate. We consider here the dynamics of linear deviations about thermodynamic equilibrium. Hence, let $\varphi^A(t,x^j)$  be some real linear-order perturbation fields that we use to characterize the system's state. Thus, $\varphi^A=0$ in equilibrium.  We now prove the following result: 
\begin{theorem}\label{theo}
Consider the following system of partial differential equations on $\mathbb{R}^{1+3}$,
\begin{equation}\label{prima}
M^\mu_{AB}\partial_\mu \varphi^B = -\Xi_{AB} \varphi^B \, ,
\end{equation}
where $M^\mu_{AB}$ and $\Xi_{AB}$ are constant
matrices, and $M^0_{AB}$ is invertible. Suppose that there exist constant symmetric matrices $E^\mu_{AB}$ and $\sigma_{AB}$ such that equation \eqref{Emu} holds for any smooth solution of \eqref{prima}, 
where $E^\mu$ and $\sigma$ are given by
\begin{equation}\label{second}
 E^\mu= \dfrac{1}{2} E^\mu_{AB} \varphi^A \varphi^B\,, \, \qquad  
 \sigma =  \sigma_{AB} \varphi^A \varphi^B \,,
\end{equation}
and $E^\mu$ is future-directed timelike (and hence non-vanishing) over the support of $\varphi^A$. Then, the system \eqref{prima} is causal, and it can be equivalently rewritten in a symmetric hyperbolic form as follows:
\begin{equation}\label{cinque}
E^\mu_{AB} \partial_\mu \varphi^B = -\sigma_{AB} \varphi^B - \Xi_{[AB]}\varphi^B \,. 
\end{equation} 
\end{theorem}


 
 

\begin{proof}
System \eqref{prima} admits smooth solutions of the form
\begin{equation}\label{thevictory}
    \varphi^A(t,x^j)=  \big(e^{-(M^0)^{-1} \Xi t} \big)\indices{^A_B} \, (Z^B-W^B) 
    + \big(e^{-(M^0)^{-1} (\Xi+M^j a_j) t} \big)\indices{^A_B} \, W^Be^{a_j x^j} \, , 
\end{equation}
for any real $Z^A,W^A,a_j=\text{const}$. Thus, $E^\mu_{AB}Z^AZ^B= 2E^\mu(0)$ must be timelike future directed for any $Z^A \neq 0$. In particular, $E^0_{AB}$ must be positive definite and, hence, invertible. Now, we always have the freedom to redefine the matrices $M^\mu_{AB}$ and $\Xi_{AB}$ by contracting both sides of \eqref{prima} with an invertible matrix $\mathcal{N}^A_C$. Let us then ``fix'' the matrix $M^0_{AB}$ to coincide with $E^0_{AB}$ (this is possible because both are invertible). If we contract \eqref{prima} with $\varphi^A$, and we plug \eqref{second} into \eqref{Emu}, we obtain the two equations below:
\begin{equation}\label{gringo}
\varphi^A M^\mu_{AB}\partial_\mu \varphi^B + \varphi^A \Xi_{AB} \varphi^B =0 \, , \spc  \varphi^A E^\mu_{AB} \partial_\mu \varphi^B + \varphi^A \sigma_{AB} \varphi^B =0 \, . 
\end{equation}
Both are respected along \textit{all} solutions of \eqref{prima}. If we subtract the second equation of \eqref{gringo} to the first, the terms $\varphi^A M^0_{AB}\partial_t \varphi^B$ cancel out (we have fixed $M^0_{AB}=E^0_{AB}$). Evaluating the result along \eqref{thevictory}, at $x^\mu=0$, we obtain
\begin{equation}\label{MjjjEjjj}
Z^A(M^j_{AB}-E^j_{AB})a_j W^B +Z^A (\Xi_{AB}- \sigma_{AB} ) Z^B =0 \, .
\end{equation}
Since this must be true for any choice of $Z^A$, $W^A$, and $a_j$, we obtain $M^j_{AB}=E^j_{AB}$ and $\Xi_{(AB)}= \sigma_{AB}$. We have recovered \eqref{cinque}. But the matrices $E^\mu_{AB}$ are symmetric \footnote{Note that, if in \eqref{second} we did not define $E^\mu_{AB}$ to be symmetric, in the second equation of \eqref{gringo} there would be $E^\mu_{(AB)}$ in place of $E^\mu_{AB}$, and we would eventually get $M^\mu_{AB}=E^\mu_{(AB)}$. Hence, the system would anyway be symmetric.}, and $E^0_{AB}$ is positive definite. Thus, the system \eqref{cinque} is symmetric hyperbolic and, since $E^\mu_{AB}Z^AZ^B$ is future-directed timelike for any $Z^A \neq 0$, it is also causal \cite{Geroch_Lindblom_1991_causal}.
\end{proof}
This theorem  
shows that thermodynamic stability in relativistic systems implies not only causality \cite{GavassinoCausality2021} but also symmetric hyperbolicity in the linear regime (provided that we have an information current). This is convenient given that, in most physical systems, symmetric-hyperbolicity follows directly from Onsager symmetry \cite{GavassinoCasmir2022}. Furthermore, by showing that the EOM are symmetric hyperbolic, we establish  that the initial value problem of all thermodynamically consistent theories 
linearized about homogeneous equilibrium
is well-posed   \cite{GerochPartial1996,Kato2009,CourantHilbert2_book}, i.e., given initial data, solutions to the equations always exist, are unique, and depend continuously on the data. 
The assumptions in Theorem \ref{theo} are quite general, encompassing an astounding number of different theories. For example, they are satisfied by all Israel-Stewart-like theories \cite{Israel_Stewart_1979,Hishcock1983,Baier2008,DMNR2012}, in an arbitrary hydrodynamic frame \cite{OlsonLifsh1990,BritoStability2020,
NoronhaGeneralFrame2021}, and with an arbitrary number of chemical species \cite{Almaalol2022}. They are also satisfied within Carter's multifluid theory \cite{carter1991,noto_rel,
GavassinoStabilityCarter2022}, GENERIC-based theories \cite{Otting2_1999,Otting3_1999,Stricker2019,
GavassinoGENERIC2022}, Geroch-Lindblom theories \cite{Geroch_Lindblom_1991_causal,LindblomRelaxation1996,
Geroch1995,GavassinoNonHydro2022}, and divergence-type theories \cite{Liu1986,GerochLindblom1990,
ZanniasDivergenceType2023}. Theorem \ref{theo} shows that the EOM of all of those theories can be found (and written in symmetric-hyperbolic form) in terms of the information current. Later in this paper we show how the information current can be systematically determined from symmetry arguments. In the Supplemental Material, we explore the origin of Theorem \ref{theo} in detail.

Finally, we remark that the hypotheses of Theorem \ref{theo} are violated by first-order theories \cite{Bemfica2017TheFirst,
Bemfica2019_conformal1,Bemfica-Disconzi-Graber-2021,
BemficaDNDefinitivo2020,Kovtun2019,GavassinoLyapunov_2020,DoreGavassino2022} because their regularized information current contains derivatives \cite{GavassinoInfoBDNK2024ufs}. Also, we note that in most hydrodynamic theories, apart from the ``holographic theories'' discussed in \cite{Heller2014,GavassinoNonHydro2022}, one finds $\Xi_{[AB]}=0$. Consequently, in this case, the linear field equations \eqref{cinque} are uniquely determined by the information current and the entropy production rate. We will assume this in the following and work with theories specified by the triplet $\{\varphi^A,E^\mu,\sigma \}$. 

In this context,  one can find conditions under which two seemingly different theories are, in reality, different manifestations of the same near-equilibrium physics. This is the content of:  
\begin{theorem}\label{theo2}
Let $\{\varphi^A, E^\mu,\sigma \}$ and $\{\tilde{\varphi}^C, \tilde{E}^\mu,\tilde{\sigma} \}$ be two linear theories, for which all the hypotheses of Theorem \ref{theo} hold, and such that $\Xi_{[AB]}=\tilde{\Xi}_{[CD]}=0$. Then, such theories are equivalent if and only if there is an invertible matrix $\mathcal{N}^A_C$ such that, for arbitrary $Z^C$,
\begin{equation}\label{potato}
E^\mu(\mathcal{N}^A_C Z^C)=\tilde{E}^\mu(Z^C) \,,\qquad   \sigma(\mathcal{N}^A_C Z^C)= \tilde{\sigma}(Z^C) \,.
\end{equation}
\end{theorem}
\begin{proof}
Theorem \ref{theo} implies that the field equations of the two theories can be recast as
\begin{equation}\label{cinquantaAcinque}
E^\mu_{AB} \partial_\mu \varphi^B = -\sigma_{AB} \varphi^B \, ,\qquad 
\tilde{E}^\mu_{CD} \partial_\mu \tilde{\varphi}^D = -\tilde{\sigma}_{CD} \tilde{\varphi}^D \, .
\end{equation}
Suppose that an invertible matrix $\mathcal{N}^A_C$ that satisfies \eqref{potato} exists. Then $E^\mu_{AB}\mathcal{N}^A_C \mathcal{N}^B_D=\tilde{E}^\mu_{CD}$, and $\sigma_{AB}\mathcal{N}^A_C \mathcal{N}^B_D=\tilde{\sigma}_{CD}$. But this implies that if we contract the first equation of \eqref{cinquantaAcinque} with $\mathcal{N}^A_C$, and we make the replacement $\varphi^B=\mathcal{N}^B_D \tilde{\varphi}^D$, we obtain the second equation of \eqref{cinquantaAcinque}. Hence, the equations of the two theories are the same, just written using different variables. Vice versa, suppose that the two linear theories are the same theory. Then, since the information current is unique \cite{GavassinoCausality2021}, we must have that $E^\mu = \tilde{E}^\mu$, and $\sigma=\tilde{\sigma}$. But this is equivalent to saying that there is a one-to-one mapping $\varphi^A =\mathcal{N}^A_C \tilde{\varphi}^C$ for which equation \eqref{potato} holds, with $Z^C=\tilde{\varphi}^C$.
\end{proof}

Theorem \ref{theo2} can be employed to prove that many apparently different theories currently in use reduce to exactly the same theory close to equilibrium. For example, a fluid mixture of two chemical substances undergoing a chemical reaction is indistinguishable from the Israel-Stewart theory for bulk viscosity \cite{Causality_bulk}, close to equilibrium (see also \cite{BulkGavassino,Camelio2022,CamelioSimulations2022}). Several other examples are discussed in our companion paper.\\

\section{Israel-Stewart representation}
In classical field theory, the existence of conservation laws is associated with the presence of a collection of currents $j_I^\mu$ (where $I$ is a new multi-index spanning the conserved quantities) with vanishing divergence $\partial_\mu j_I^\mu=0$ \cite{Peskin_book} (e.g., baryon number conservation). In linear fluid theories characterized by $\{\varphi^A,E^\mu,\sigma\}$ such as those considered here, such currents manifest themselves through the existence of a (constant) matrix $\mathcal{N}^A_I$ such that $\mathcal{N}^A_I \sigma_{AB}=0$.
In fact, if we contract the field equations of the theory, $E^\mu_{AB}\partial_\mu \varphi^B=-\sigma_{AB}\varphi^B$, with $\mathcal{N}^A_I$, the right-hand side vanishes, and we recover the equations $\partial_\mu j_I^\mu=0$, with
\begin{equation}\label{martina}
    j_I^\mu = \mathcal{N}^A_I E^\mu_{AB} \varphi^B \, .
\end{equation}
In the Supplementary Material, we prove the following useful result:
\begin{theorem}\label{theoLind}
Let $\{\varphi^A,E^\mu,\sigma \}$ be a linear theory for which all the hypotheses of Theorem \ref{theo} hold, and $\Xi_{[AB]}=0$. If the  conservation laws $\partial_\mu ( \mathcal{N}^A_I E^\mu_{AB} \varphi^B)=0$ are all independent, then there is a one-to-one change of variables $\varphi^A \rightarrow \{ \mu^I,\Pi^a \}$ such that $E^\mu$, $j_I^\mu$, and $\sigma$ take the form (all matrices below are constant)
\begin{equation}\label{decompose}
    E^\mu = \dfrac{1}{2} E^\mu_{IJ} \, \mu^I \mu^J + E^\mu_{Ib} \, \mu^I \Pi^b +\dfrac{1}{2} E^\mu_{ab} \, \Pi^a \Pi^b \, ,\qquad \qquad 
    j^\mu_I =  E^\mu_{IJ} \,  \mu^J + E^\mu_{Ib} \, \Pi^b  \, , \qquad \qquad
    \sigma = \sigma_{ab} \, \Pi^a \Pi^b \, ,  \\
\end{equation}
where $E^\mu_{IJ}$, $E^\mu_{ab}$, and $\sigma_{ab}$ are symmetric matrices. If $\mathcal{N}^A_I$ accounts for all conservation laws, then $\sigma_{ab}$ is invertible. 
\end{theorem}
Theorem \ref{theoLind} tells us that one can always rearrange the EOM so that their mathematical structure ``resembles'' Israel-Stewart theory. In fact, if $\sigma_{ab}$ is invertible, with matrix inverse $\sigma^{ab}$, we can express the field equations \eqref{cinque} in terms of the variables $\{ \mu^I,\Pi^a \}$ as
\begin{equation}\label{lindblomrepr}    
    \partial_\mu (E^\mu_{IJ} \,  \mu^J + E^\mu_{Ib} \, \Pi^b)= 0 \, ,\qquad
       \sigma^{ab} E^\mu_{bc} \,\partial_\mu \Pi^c +  \,\Pi^a  = -  \sigma^{ab} E^\mu_{Jb} \,\partial_\mu \mu^J \, ,
\end{equation}
which we will refer to as the ``Israel-Stewart representation'' of the theory. The first set of equations in \eqref{lindblomrepr} is the set of all conservation laws. The second set gives relaxation-type equations \cite{Jou_Extended}, which describe dissipation. Stability requires that $E^0_{ab}$ and $\sigma_{ab}$ be positive definite so that the second equation in \eqref{lindblomrepr} contains both $\partial_t \Pi^a$ and $\Pi^a$, breaking time-reversal invariance. Thus, one can interpret $\Pi^a$ as ``dissipative fields'' \cite{LindblomRelaxation1996} (e.g., viscous stresses, diffusive currents, and reaction affinities \cite{PrigoginebookModernThermodynamics2014}). The fields $\mu^I$ are the usual ``dynamical fluid fields'' \cite{LindblomRelaxation1996} (e.g., temperature, chemical potential, and flow velocity). Indeed, if $\Pi^a=0$ then  $2E^\mu=\mu^I j_I^\mu$, meaning that $\mu^I$ may be interpreted as the ``chemical potential'' of the conserved density  $j_I^0$   \cite{GavassinoCasmir2022,Pathria2011}. 
Therefore, thermodynamically stable relativistic theories can always be expressed in this Israel-Stewart representation specified by a choice of $\{ \mu^I,\Pi^a,E^\mu,\sigma \}$.


Further insight can be obtained by realizing that, due to rotational invariance of the equilibrium state, we may further decompose $\mu^I$ and $\Pi^a$ into irreducible tensors of the $SO(3)$ rotation group. For example, if among the fields there is a 4-current $\delta n^\mu$, then $\delta n^0$ behaves as a scalar under rotations, and $\delta n^j$ ($j=1,2,3$) behaves as a 3-vector. Hence, we can assign to a given theory a list of integers $(\mathfrak{g}_0,\mathfrak{g}_1,\mathfrak{g}_2,...)$ that specifies the geometric character of its degrees of freedom,  with  $\mathfrak{g}_0$ being the number of scalars, $\mathfrak{g}_1$ the number of vectors, $\mathfrak{g}_2$ the number of symmetric traceless tensors with two indices, and so forth. One can repeat the same procedure for the conservation laws, including in the count only the fields $\mu^I$, which transform under $SO(3)$ as the respective conserved densities $j_I^0$. This produces a second list of integers,  $(\bar{\mathfrak{g}}_0,\bar{\mathfrak{g}}_1,\bar{\mathfrak{g}}_2,...)$, where $\bar{\mathfrak{g}}_n \leq \mathfrak{g}_n$. From \eqref{lindblomrepr}, we see that the theory is non-dissipative if and only if $\bar{\mathfrak{g}}_n =\mathfrak{g}_n$, $\forall \, n$, i.e., if there are as many conservation laws as degrees of freedom. Therefore, one can fully specify a given theory by saying that it is ``of class $(\mathfrak{g}_0,\mathfrak{g}_1,\mathfrak{g}_2,...)-(\bar{\mathfrak{g}}_0,\bar{\mathfrak{g}}_1,\bar{\mathfrak{g}}_2,...)$''. 
For example, a perfect fluid at finite chemical potential is of class $(2,1)-(2,1)$, and it does not dissipate, whereas a bulk-viscous fluid at zero chemical potential is of class $(2,1)-(1,1)$, and it dissipates (here, it is understood that $\mathfrak{g}_n=\bar{\mathfrak{g}}_n=0$, for $n \geq 2$). 

\section{Universality classes}
Pick a class, i.e. fix the values of $\mathfrak{g}_n$ and $\bar{\mathfrak{g}}_n$. Working in the Israel-Stewart representation, give some names to the fields $\mu^I$ and $\Pi^a$ and construct the most general expressions for $E^\mu$ and $\sigma$, of the form \eqref{decompose}, compatible with rotational invariance. This produces the most general theory of the given class. By Theorem \ref{theo2}, any other theory belonging to the same class must be a particular case of this general theory (for some specific choice of parameters) since a mapping of the form \eqref{potato} is guaranteed to exist ($E^\mu$ and $\sigma$ being the most general). 
This general theory is usually very complicated, as it possesses a plethora of free coefficients. Luckily, Theorem \ref{theo2} comes to our aid: one can reabsorb many transport coefficients through changes of variables, as this does not modify the dynamics of the system. A useful type of field redefinition is the ``change of hydrodynamic frame'': $\mu^I = \tilde{\mu}^I +\mathcal{R}^I_c \tilde{\Pi}^c$ and $\Pi^a = \mathcal{R}^a_c \tilde{\Pi}^c$ ($\mathcal{R}^I_c$ and $\mathcal{R}^a_c$ being constant matrices), which preserves the structure \eqref{decompose},  mapping Israel-Stewart representations into  Israel-Stewart representations. The goal is to find a transformation that maps the general theory into an already existing theory (whose physical interpretation is known), which plays the role of a ``representative'' of the class. If this happens, Theorem \ref{theo2} guarantees that any linear theory belonging to the class is a particular realization of the representative and exhibits the same physical behavior.

We have applied this method to some selected (parity invariant \cite{SperanzaChiral2021}) classes. The representatives are reported in Table \ref{tableI} (we have fixed $\mathfrak{g}_n=\bar{\mathfrak{g}}_n=0$, for $n>2$). An interesting pattern can be recognized. In the absence of vector conservation laws, the currents do not have inertia, and they can only diffuse. If we include one vector conservation law, which plays the role of the linear momentum, only then will the system behave like an actual fluid. If we include more vector conservation laws, the fluid can sustain multiple non-diffusive relative flows, and it behaves like a superfluid. If we include a tensor conservation law, the system can conserve the ``memory'' of the deformations it experiences, and it becomes elastic. Combine two vector conservation laws with one tensor conservation law, and the result is a superfluid-elastic system, i.e., a supersolid.
\begin{table}[h!]
\centering
\begin{tabular}{ | c |c c c| l| }
\hline
$\mathfrak{g}_0  \, \mathfrak{g}_1 \, \mathfrak{g}_2 $&$\bar{\mathfrak{g}}_0$ & $\bar{\mathfrak{g}}_1$ & $\bar{\mathfrak{g}}_2$ & Representatives \\
 \hline
 \hline
  $a \, \, \, \, 0 \, \,  \,\, 0$ & $\leq a$ & 0 & 0 & Chemistry \\ 
 \hline
  $a \, \, \, \, 1 \, \,  \,\, 0$ & $\leq a$ & 1 & 0 & Fluid mixture \cite{noto_rel,BulkGavassino}; Models for bulk viscosity \cite{Causality_bulk} \\ 
 \hline
$a \, \, \, \, a \, \,  \,\, 0$ &  $\leq a$ & $ \leq a$ & 0 & Carter multifluids \cite{Carter_starting_point,GavassinoStabilityCarter2022} \\ 
 \hline
  \multirow{3}{3em}{ $1 \, \, \, \, 1 \, \,  \,\, 0$ }  & 0 & 0 & 0 & Diffusion of a non-conserved density \\ 
    & 1 & 0 & 0 & Cattaneo model of diffusion \cite{cattaneo1958,Jou_Extended} \\ 
      & 1 & 1 & 0 & Perfect fluid at $\mu=0$; Barotropic perfect fluid \\ 
 \hline
  \multirow{2}{3em}{ $2 \, \, \, \, 1 \, \,  \,\, 0$ }  & 1 & 1 & 0 & Bulk viscous fluid at $\mu=0$ \\ 
      & 2 & 1 & 0 & Perfect fluid  \\ 
\hline
  \multirow{4}{3em}{ $2 \, \, \, \, 2 \, \,  \,\, 0$ } & 2 & 0 & 0 & Coupled diffusion of two conserved densities \cite{Onsager1931} \\ 
  & 1 & 1 & 0 & Heat conductive bulk viscous fluid at $\mu=0$ \\ 
    & 2 & 1 & 0 & Heat conductive fluid at $\mu \neq 0$ \cite{OlsonRegularCarter1990} \\ 
      & 2 & 2 & 0 & Relativistic superfluid \cite{cool1995} \\ 
 \hline
   \multirow{2}{3em}{ $1 \, \, \, \, 1 \, \,  \,\, 1$ }  & 1 & 1 & 0 & Maxwell material \cite{BaggioliHolography,GavassinoGENERIC2022} at $\mu=0$ \\ 
      & 1 & 1 & 1 & Elastic material at $\mu=0$ or at $T=0$ \cite{Schumaker1983,landau7}   \\ 
 \hline
$1 \, \, \, \, 1 \, \,  \,\, 2$ &  1 & 1 & 0 & Burgers material \cite{Malek2018} at $\mu=0$; MIS$^*$ \cite{KeYin2022} \\ 
 \hline
\multirow{4}{3em}{$3 \, \, \, \, 2 \, \,  \,\, 1$} &  1 & 1 & 0 & Israel-Stewart theory in a ``general frame'' \cite{NoronhaGeneralFrame2021} at $\mu = 0$  \\ 
&  2 & 1 & 0 & Israel-Stewart theory \cite{Israel_Stewart_1979,Hishcock1983}  at $\mu \neq 0$ \\ 
&  3 & 1 & 1 & Elastic heat conducting material  \\ 
&  3 & 2 & 1 & Supersolid \cite{Andreev1969,Sears2010}; Inner crust of neutron stars \cite{CarterSamuelsson2006}  \\ 
 \hline
\end{tabular}
\caption{Universality classes near equilibrium.}
\label{tableI}
\end{table}
We provide the information currents and entropy production rates of the theories listed in Table \ref{tableI} in the Supplemental Material. In our companion paper, we use those expressions to discuss in detail the equivalence between the theories in the most relevant classes. We consider here
a system of class $(1,1,1)-(1,1,1)$. Its fields are $\mu^I=\{\delta \mu, \delta u^k,\delta \Pi^{kl} \}$, which may be interpreted as the perturbations to the chemical potential, flow velocity, and shear stress tensor (which is symmetric and traceless). Since we have as many degrees of freedom as conservation laws (i.e., there are no fields of ``$\Pi^a$ type''), the entropy production rate vanishes. Thus, the most general theory is
\begin{equation}\label{infoelas}
    TE^0 =\dfrac{1}{2}  \dfrac{d n}{d\mu} (\delta \mu)^2 + \dfrac{1}{2} (\rho {+} P)\delta u^k \delta u_k + \dfrac{\delta \Pi^{kl}\delta \Pi_{kl}}{4G} \, ,\qquad
    TE^j = n\delta \mu \delta u^j + \delta \Pi^{jk} \delta u_k \, ,\qquad
    T\sigma = 0 \,.
\end{equation}
This information current contains all possible terms allowed by symmetry, making this theory a valid representative of its class. The transport coefficients have been given a physically-motivated name to ease their interpretation: $\rho {+} P$ may be interpreted as the enthalpy density, $G$ as the shear modulus, and $n$ as the background density. The prefactor of $\delta \Pi^{jk} \delta u_k$ could be set to unity by appropriately rescaling $\delta \Pi^{jk}$ with a field redefinition of the form $\delta \Pi^{jk} \rightarrow a \, \delta \Pi^{jk}$\footnote{The case in which the prefactor of $\delta \Pi^{jk} \delta u_k$ equals zero is trivial since it would produce the field equation $\partial_t \delta \Pi_{kl}=0$.}. If we compute the field equations from equation \eqref{cinque}, we obtain
\begin{equation}\label{dayofrain}
     \partial_t \delta n + \partial_j (n \delta u^j)=0 \, , \qquad
     (\rho {+} P)\partial_t \delta u_k + \partial_k (n \delta \mu) + \partial_j \delta \Pi^j_k =0 \, , \qquad
     \dfrac{\partial_t \delta \Pi_{kl}}{2G} + \braket{\partial_k \delta u_l}=0 \, , 
\end{equation}
where $\braket{...}$ extracts the symmetric traceless part. To obtain the last equation,  one needs to account for the constraints on $\delta \Pi_{kl}$ (see \cite{GavassinoCasmir2022} for a detailed discussion). Equations \eqref{dayofrain} describe the dynamics of an isotropic elastic material at zero temperature. The first equation is the continuity equation for particles, the second is the conservation of momentum, the third incorporates shear-stress dynamics in the Hookean approximation \cite{Schumaker1983}. Combining all three equations, and using $dP/d\rho = n^2 d\mu/(\rho {+} P)dn$, we obtain
\begin{equation}\label{anotherdayofsun}
    \partial^2_t \delta u_k -  \dfrac{G}{\rho {+} P} \partial^j \partial_j \delta u_k -\bigg[\dfrac{dP}{d\rho} + \dfrac{G}{3(\rho {+} P)} \bigg] \partial_k \partial_j \delta u^j =0 \, ,
\end{equation}
which is consistent with standard formulas from the theory of elasticity \cite{landau7}. This describes an elastic material that can sustain both longitudinal and transversal propagating waves with speeds $c_L=\sqrt{\dfrac{dP}{d\rho} + \dfrac{4G}{3(\rho {+}P)}}$ and $c_T= \sqrt{\dfrac{G}{\rho {+} P}}$, respectively \cite{landau7}.


\section{Conclusions}
Our results establish, for the first time, connections between elasticity, viscosity, superfluidity, supersolidity, diffusion, and chemistry within a systematic, fully relativistic formalism. Using this paper, one can easily write down the (linearized) EOM of an arbitrary relativistic hydrodynamic system with very limited knowledge about its behavior: one only needs to know the relevant degrees of freedom and conservation laws \footnote{We note that, in general, knowledge about dispersion relations is not enough to specify the dynamics of a system uniquely. In fact, causality cannot solely be determined from dispersion relations, see \cite{Gavassino:2023mad} and our companion paper. On the other hand, within our approach, given the couple $\{E^\mu,\sigma\}$, the equations of motion are unique, and causal by construction. Furthermore, when one specifies the numbers $(\mathfrak{g}_0,\mathfrak{g}_1,\mathfrak{g}_2,...)-(\bar{\mathfrak{g}}_0,\bar{\mathfrak{g}}_1,\bar{\mathfrak{g}}_2,...)$, the transformation laws of the fields and possible constraints (such as $\Pi_{jk}=\Pi_{kj}$) are implemented from the start. Hence, our method, in general, provides more information about the system's dynamics than a pure spectral analysis.}. Then, the most general information current can be constructed, from which the EOM can be derived in full analogy with action principles. Our method's three main advantages are: (a) The resulting EOMs are always causal, stable, thermodynamically consistent, and uniquely solvable for smooth initial data. (b) Given the degrees of freedom and the conservation laws, the associated theory is unique. (c) The information current uniquely determines also the fluctuating generalization of the theory, giving rise to well-posed stochastic dynamics in terms of a path integral \cite{Mullins:2023tjg}.

\section{Acknowledgements}
M.M.D. is partially supported by NSF grant DMS-2107701, a Vanderbilt's Seeding Success Grant, 
a Chancellor's Faculty Fellowship, and DOE grant DE-SC0024711.
L.G. is partially supported by a Vanderbilt's Seeding Success Grant. J.N. is partially supported by the U.S. Department of Energy, Office of Science, Office for Nuclear Physics
under Award No. DE-SC0023861.

\bibliography{Biblio}

\onecolumngrid
\newpage
\begin{center}
{\bf \large SUPPLEMENTARY MATERIAL}
\end{center}

\setcounter{equation}{0}
\setcounter{figure}{0}
\setcounter{table}{0}
\setcounter{page}{1}
\renewcommand{\theequation}{S\arabic{equation}}
\renewcommand{\thefigure}{S\arabic{figure}}

\maketitle

\section{Origin and properties of the information current}\label{sectionI}

Here, we briefly review how the information current is explicitly computed for an assigned (fully non-linear) theory. The foundations of the method where laid out in \cite{GavassinoGibbs2021} and refined in subsequent works \cite{GavassinoCausality2021,
GavassinoStabilityCarter2022,GavassinoGENERIC2022}. 

The equilibrium statistical operator of a fluid in contact with a large environment is $\hat{\rho}\propto \exp(\alpha_\star^I \hat{Q}_I)$, where $\hat{Q}_I$ are the Noether charges (energy, momentum, electric charge,...) of the underlying microscopic theory, and $\alpha_\star^I$ are some constants which are uniquely determined by the thermodynamic state of the environment. The operator $\hat{\rho}$ defines the grand-canonical ensemble of a system with arbitrary conservation laws \cite{Hakim2011}. Let $\psi^A$ be the hydrodynamic fields of the (fully non-linear) hydrodynamic theory, and consider a specific field configuration $\psi^A(\Sigma)$  on a given Cauchy surface $\Sigma$. This configuration constitutes a possible choice of initial data for the system, and defines a macroscopic state, with many possible microscopic realizations. Let $\hat{\mathcal{P}}[\psi^A(\Sigma)]$ be the projector onto the Hilbert subspace of all the possible microscopic realizations of $\psi^A(\Sigma)$. The entropy of $\psi^A(\Sigma)$ is given by Boltzmann's rule: $\exp S[\psi^A(\Sigma)]=\text{Tr} \hat{\mathcal{P}}[\psi^A(\Sigma)] $. Therefore, the equilibrium probability of observing the hydrodynamic state $\psi^A(\Sigma)$ is
\begin{equation}
\mathcal{P} \propto \text{Tr}\big(\hat{\mathcal{P}}[\psi^A(\Sigma)]\exp(\alpha_\star^I \hat{Q}_I) \big) \approx \text{Tr}\big(  \hat{\mathcal{P}}[\psi^A(\Sigma)] \big) \exp(\alpha_\star^I Q_I[\psi^A(\Sigma)]) = \exp(S[\psi^A(\Sigma)]+\alpha_\star^I Q_I[\psi^A(\Sigma)]) \, .
\end{equation}
Clearly, the equilibrium macrostate is the most probable macrostate. We can therefore conclude that the equilibrium state of the hydrodynamic theory must be the state which maximizes (for unconstrained variations) the functional
\begin{equation}
\Phi[\psi^A]= S[\psi^A] + \alpha_\star^I Q_I[\psi^A] \, ,
\end{equation} 
for fixed values of all $\alpha_\star^I$ (which depend on the environment). Then, assuming that the theory has an entropy current $s^\mu(\psi^A)$, we need to maximize the flux-integral (we adopt the ``standard orientation'' \cite{MTW_book}: $d\Sigma_0 >0$)
\begin{equation}\label{phiuzzo}
\Phi= \int_\Sigma (s^\mu +\alpha_\star^I J^\mu_I)d\Sigma_\mu  \, ,
\end{equation}
where $J^\mu_I(\psi^A)$ are the conserved currents of the theory. In order to actually compute the maximum of \eqref{phiuzzo}, we introduce a smooth one-parameter family, $\psi^A(\epsilon)$, of solutions of the field equations such that $\epsilon=0$ is the equilibrium state. Then, we impose $\dot{\Phi}(\epsilon=0)=0$ for any possible one-parameter family of this kind (we adopt the notation $\dot{f}:=df/d\epsilon$). This immediately leads to the covariant Gibbs relation \cite{Israel_2009_inbook},
\begin{equation}\label{Gibbsusss}
\dfrac{\partial s^\mu}{\partial \psi^A} = -\alpha_\star^I \dfrac{\partial J^\mu_I}{\partial \psi^A} \, ,
\end{equation}
which fixes the equilibrium state uniquely. Then, if we define $E:=\Phi(0)-\Phi(\epsilon)$, the condition that $\Phi$ should be maximized in equilibrium is equivalent to the requirement that $E$ should be strictly positive for non-vanishing perturbations. Performing a second-order Taylor expansion in $\epsilon$, and recalling equation \eqref{Gibbsusss}, we obtain\footnote{Note that the terms proportional to $\ddot{\psi}^A$ in $E$ cancel out as a consequence of \eqref{Gibbsusss}. }
\begin{equation}
E = \int_\Sigma -\dfrac{1}{2}\bigg[\dfrac{\partial^2 s^\mu}{\partial \psi^A \partial \psi^B} + \alpha_\star^I\dfrac{\partial^2 J^\mu_I}{\partial \psi^A \partial \psi^B}  \bigg] \varphi^A \varphi^B d\Sigma_\mu +\mathcal{O}(\epsilon^3) \geq 0  \, .
\end{equation}
The quantity inside the integral is $E^\mu$, and we are adopting the notation $\varphi^A:= \psi^A(\epsilon)-\psi^A(0)=\dot{\psi}^A(0) \epsilon+\mathcal{O}(\epsilon^2)$. If we require $E$ to be positive-definite for all choices of $\Sigma$, we find that $E^\mu$ is future-directed timelike. This automatically enforces hydrodynamic stability (in the linear regime), if the theory is consistent with the second law of thermodynamics. To see this, consider that, if no interaction occurs between the fluid and the environment, then $\Delta S \geq 0$ and $\Delta Q_I=0$ along any hydrodynamic process. Therefore, $\Delta \Phi=\Delta S \geq 0$ (recall that $\alpha_\star^I$ are constants). It follows that $\Phi(\epsilon)$ can only increase, or stay constant. Considering that $\Phi(0)$ is the equilibrium value of $\Phi$ (which is a constant), we can conclude that $E=\Phi(0)-\Phi(\epsilon)$ is non-increasing in time. In summary, $E$ plays the role of a positive-definite square-integral norm of the perturbations, which is bounded from below by zero and from above by its initial value. This guarantees (Lyapunov) stability of the equilibrium and global uniqueness of the solution for square-integrable initial data, in the linear regime \cite{Wald}. With a simple Gauss-theorem argument, one can also prove causality \cite{GavassinoCausality2021}.

\newpage

\section{Understanding Theorem 1} 

Theorem 1 is quite a surprising result, as it tells us that the very existence of the information current fully constrains the mathematical structure of the field equations. How is it possible? Here, we present an alternative proof of Theorem 1, which should help the reader understand what is really going on. We also provide an explicit example.

\subsection{Proof from linearity}

It is possible to prove Theorem 1 without ever manipulating the field equations ``$M^\mu_{AB}\partial_\mu \varphi^B = -\Xi_{AB}\varphi^B$'' explicitly, but only relying on their linearity. Let us see how this works. Consider two arbitrary solutions, $\varphi^A$ and $\tilde{\varphi}^A$, of the field equations. Then, since $\partial_\mu E^\mu =-\sigma$ holds along all solutions of the field equations, we must have\footnote{Note that, if we did not define $E^\mu_{AB}$ to be symmetric, in equation \eqref{pippo} there would be $E^\mu_{(AB)}$ in place of $E^\mu_{AB}$.} (recalling that $2E^\mu=E^\mu_{AB}\varphi^A \varphi^B$ and $\sigma=\sigma_{AB}\varphi^A \varphi^B$)
\begin{equation}\label{pippo}
\begin{split}
& \varphi^A E^\mu_{AB}\partial_\mu \varphi^B +\varphi^A \sigma_{AB}\varphi^B =0 \, , \\
& \tilde{\varphi}^A E^\mu_{AB}\partial_\mu \tilde{\varphi}^B +\tilde{\varphi}^A \sigma_{AB}\tilde{\varphi}^B =0 \, . \\
\end{split}
\end{equation}
Furthermore, since the equations are linear, also $\varphi^A+  \tilde{\varphi}^A$ must be a solution of the field equations, and $\partial_\mu E^\mu =-\sigma$ (which is obeyed by \textit{all} solutions) must hold also along $\varphi^A+  \tilde{\varphi}^A$. Therefore:
\begin{equation}\label{varuzzo}
(\varphi^A+ \tilde{\varphi}^A) E^\mu_{AB}\partial_\mu (\varphi^B+ \tilde{\varphi}^B) +(\varphi^A+ \tilde{\varphi}^A) \sigma_{AB}(\varphi^B+ \tilde{\varphi}^B) =0 \, .
\end{equation}
Expanding all the products, and using \eqref{pippo} to cancel some terms, we obtain
\begin{equation}\label{iotiuso}
\tilde{\varphi}^A E^\mu_{AB}\partial_\mu \varphi^B +\varphi^A E^\mu_{AB}\partial_\mu \tilde{\varphi}^B +2 \, \tilde{\varphi}^A \sigma_{AB}\varphi^B  =0 \, .
\end{equation}
Now we immediately see what is happening: since $E^\mu$ is quadratic in the fields, and the second law $\partial_\mu E^\mu =-\sigma$ must hold for all linear combinations of solutions, there are some ``cross-conditions'' on all couples of solutions $\{\varphi^A, \tilde{\varphi}^A \}$, and this induces very strong constraints on the structure of the field equations. For example, if at some point $p$ we have $\varphi^A(p)=0$, then equation \eqref{iotiuso} implies that $E^\mu_{AB}\partial_\mu \varphi^B=0$ at $p$. We can already see that the matrices $E^\mu_{AB}$ must be somehow related to the matrices $M^\mu_{AB}$ defining the principal part of the field equations. Our goal now is to show that we can fully ``reconstruct'' the field equations of the theory just using \eqref{iotiuso}. The idea is the following.

Suppose that the multi-index $A$ runs from $1$ to $\mathfrak{D}$. Pick an arbitrary spacetime event $p$ and construct a collection of $\mathfrak{D}$ local solutions\footnote{Note that the existence of such local solutions is guaranteed by the Cauchy-Kovalevskaya (CK) theorem \cite{Rauch_book,CourantHilbert2_book}. In fact, the invertibility of $M^0_{AB}$ guarantees that the equations can be recast in the normal form, and linear equations with constant coefficients obviously have analytic coefficients, so that the CK theorem applies. The exponential solutions that we present in the main text are indeed an example of analytic solutions. It is also easy to show that, if we release the assumption that $M^0_{AB}$ is invertible, the proof no longer works. For example, take the case $M^\mu_{AB}=0$, $\Xi_{AB}=\delta_{AB}$. Then the only solution is $\varphi^A=0$, and there exists no solution $\tilde{\varphi}^A$ such that $\tilde{\varphi}^A(p)=\delta^A_C$.}, $\{\tilde{\varphi}^A_{(C)} \}_{C=1}^\mathfrak{D}$, of the field equations, such that $\tilde{\varphi}^A_{(C)}(p)=\delta^A_C$. Then, since $\tilde{\varphi}^A_{(C)}$ are solutions of the field equations, equation \eqref{iotiuso} must hold also replacing $\tilde{\varphi}^A$ with $\tilde{\varphi}^A_{(C)}$. This produces a collection of $\mathfrak{D}$ conditions at $p$:
\begin{equation}
E^\mu_{CB}\partial_\mu \varphi^B +\varphi^A E^\mu_{AB}\partial_\mu \tilde{\varphi}^B_{(C)} +2 \, \sigma_{CB}\varphi^B  =0 \spc (\text{at }p).
\end{equation}
Defining $\mathcal{G}_{CA}:= [E^\mu_{AB}\partial_\mu \tilde{\varphi}^B_{(C)}]_p$ (the subscript $p$ means ``evaluated at $p$''), and isolating $E^\mu_{CB}\partial_\mu \varphi^B$, we obtain
\begin{equation}\label{almost}
E^\mu_{CB}\partial_\mu \varphi^B =-(2 \, \sigma_{CB}+\mathcal{G}_{CB})\varphi^B   \spc (\text{at }p).
\end{equation}
On the other hand, we know that the linearized theory is invariant under spacetime translations (because the background state is homogeneous). Hence, we can repeat the same procedure above at any other spacetime event $q$, operating a translation $\tilde{\varphi}^A_{(C)}(x)\rightarrow \tilde{\varphi}^A_{(C)}(x-q+p)$, so that now they reduce to $\delta^A_C$ at $q$. This implies that the condition \eqref{almost} must hold at every spacetime point, with the same $\mathcal{G}_{CB}$. Furthermore, if we plug \eqref{almost} into the first line of \eqref{pippo}, we find that the matrix $\sigma_{CB}+\mathcal{G}_{CB}$ must be anti-symmetric. Hence, we call it $\Xi_{[CB]}$, and \eqref{almost} finally becomes
\begin{equation}\label{finally}
E^\mu_{CB}\partial_\mu \varphi^B =-\sigma_{CB} \varphi^B -\Xi_{[CB]}\varphi^B \spc (\text{at any event }q) \, .
\end{equation}
Since $E^\mu_{CB}$ are all symmetric, and $E^\mu_{CB}\varphi^C \varphi^B=2E^\mu$ is future directed timelike $\forall \, \varphi^B\neq 0$, this is a  symmetric hyperbolic and  causal \cite{Geroch_Lindblom_1991_causal} system of $\mathfrak{D}$ partial differential equations in $\mathfrak{D}$ variables. Holmgren's uniqueness theorem \cite{Rauch_book} guarantees that such system uniquely fixes the fields for given initial data. Hence, it must necessarily contain as much information as the original field equations (which were not even invoked explicitly!). This completes our alternative proof of Theorem 1.

\subsection{Example: Deriving Maxwell's equations from  energy conservation}

Suppose that one has forgotten the two ``dynamical'' Maxwell equations in vacuum, namely $\partial_t \textbf{E}=\nabla \times \textbf{B}$, and $\partial_t \textbf{B}=-\nabla \times \textbf{E}$. However, one remembers that they are linear, and that they obey a conservation law of the form $\partial_t \rho + \nabla \cdot \textbf{S}=0$, where $\rho$ is the electromagnetic energy density and $\textbf{S}$ is the Poynting vector:
\begin{equation}
\rho = \dfrac{|\textbf{E}|^2 + |\textbf{B}|^2}{2} \, , \spc \textbf{S}= \textbf{E} \times \textbf{B} \, .
\end{equation}
Can one reconstruct the original equations? Let's see. We consider two arbitrary solutions of the (unknown) Maxwell's equations. We call them $\{\textbf{E},\textbf{B} \}$ and $\{ \tilde{\textbf{E}},\tilde{\textbf{B}} \}$. Then, the conservation law $\partial_t \rho + \nabla \cdot \textbf{S}=0$ implies
\begin{equation}\label{eormf}
\begin{split}
& \textbf{E} \cdot \partial_t \textbf{E} + \textbf{B} \cdot \partial_t \textbf{B} + \textbf{B} \cdot (\nabla \times \textbf{E})- \textbf{E} \cdot (\nabla \times \textbf{B}) =0 \, , \\
& \tilde{\textbf{E}} \cdot \partial_t \tilde{\textbf{E}} + \tilde{\textbf{B}} \cdot \partial_t \tilde{\textbf{B}} + \tilde{\textbf{B}} \cdot (\nabla \times \tilde{\textbf{E}})- \tilde{\textbf{E}} \cdot (\nabla \times \tilde{\textbf{B}}) =0 \, .\\
\end{split}
\end{equation}
This does not fix the evolution of $\textbf{E}$ and $\textbf{B}$ completely. However, let us now consider the sum $\{\textbf{E}+\tilde{\textbf{E}},\textbf{B}+\tilde{\textbf{B}} \}$. Since we know that Maxwell's equations in vacuum are linear, we know that this is also a solution. Hence, the conservation law $\partial_t \rho + \nabla \cdot \textbf{S}=0$ must hold also for the total energy density and flux:
\begin{equation}
\rho = \dfrac{|\textbf{E}+\tilde{\textbf{E}}|^2 + |\textbf{B}+\tilde{\textbf{B}}|^2}{2} \, , \spc \textbf{S}= (\textbf{E}+\tilde{\textbf{E}}) \times (\textbf{B}+\tilde{\textbf{B}}) \, .
\end{equation}
This produces the cross condition below:
\begin{equation}\label{gbuf}
\tilde{\textbf{E}} \cdot (\partial_t \textbf{E}-\nabla \times \textbf{B}) + \tilde{\textbf{B}} \cdot (\partial_t \textbf{B}+\nabla \times \textbf{E} ) +\textbf{E} \cdot (\partial_t \tilde{\textbf{E}}-\nabla \times \tilde{\textbf{B}}) + \textbf{B} \cdot (\partial_t \tilde{\textbf{B}}+\nabla \times \tilde{\textbf{E}} ) =0 \, .
\end{equation}
Now, fixed some event $p$, the values of $\tilde{\textbf{E}}$ and $\tilde{\textbf{B}}$ at $p$ may be arbitrary. Thus, if we want equation \eqref{gbuf} to hold for any choice of $\{ \tilde{\textbf{E}},\tilde{\textbf{B}} \}$, we must require that $\partial_t \textbf{E}-\nabla \times \textbf{B}$ and $\partial_t \textbf{B}+\nabla \times \textbf{E}$  be linear combinations of $\textbf{E}$ and $\textbf{B}$. Then, it is easy to see that the most general (rotationally-invariant) field equations that are consistent with \eqref{eormf} and \eqref{gbuf} are
\begin{equation}
\partial_t \textbf{E}-\nabla \times \textbf{B}= \mathcal{G}\, \textbf{B}  \, , \spc \partial_t \textbf{B}+\nabla \times \textbf{E}= -\mathcal{G} \, \textbf{E} \, ,
\end{equation}
where $\mathcal{G}$ is some constant. This a symmetric-hyperbolic system, in agreement with our theorem. As we can see, we could fully reconstruct Maxwell's equations from the condition $\partial_t \rho + \nabla \cdot \textbf{S}=0$, except for the antisymmetric coupling $\mathcal{G}$ (which corresponds to the matrix $\Xi_{[AB]}$, in our theorem). Nature has chosen $\mathcal{G}=0$.

\newpage

\section{Proof of Theorem 3}

\textbf{Theorem 3}. 
\textit{Let $\{\varphi^A,E^\mu,\sigma \}$ be a linear theory for which all the hypotheses of Theorem 1 hold, and $\Xi_{[AB]}=0$. If the  conservation laws $\partial_\mu ( \mathcal{N}^A_I E^\mu_{AB} \varphi^B)=0$ are all independent, then there is a one-to-one change of variables $\varphi^A \rightarrow \{ \mu^I,\Pi^a \}$ such that $E^\mu$, $j_I^\mu$, and $\sigma$ take the form (all matrices below are constant)}
\begin{equation}\label{decompose}
    E^\mu = \dfrac{1}{2} E^\mu_{IJ} \, \mu^I \mu^J + E^\mu_{Ib} \, \mu^I \Pi^b +\dfrac{1}{2} E^\mu_{ab} \, \Pi^a \Pi^b \, ,\qquad \qquad 
    j^\mu_I =  E^\mu_{IJ} \,  \mu^J + E^\mu_{Ib} \, \Pi^b  \, , \qquad \qquad
    \sigma = \sigma_{ab} \, \Pi^a \Pi^b \, ,  \\
\end{equation}
\textit{where $E^\mu_{IJ}$, $E^\mu_{ab}$, and $\sigma_{ab}$ are symmetric matrices. If $\mathcal{N}^A_I$ accounts for all conservation laws, then $\sigma_{ab}$ is invertible.} 
\begin{proof}
We can regard the collection of fields $\varphi^A$ as a map $\varphi^A : \mathbb{R}^{1+3}\rightarrow \mathbb{R}^{\mathfrak{D}}$, for some positive integer $\mathfrak{D}$. For a fixed value of the index $I$, we may also view $\mathcal{N}^A_I$ as an element of the vector space $\mathbb{R}^{\mathfrak{D}}$. If the conservation laws are all independent, the vectors $\mathcal{N}^A_I$ (for all $I$) form a linearly independent set. Therefore, we can complete them to a basis of $\mathbb{R}^{\mathfrak{D}}$, through the introduction of some other (constant) vectors $\mathcal{N}^A_a$. It follows that the fields $\varphi^A$ can always be expressed in the form $\varphi^A=\mathcal{N}^A_I \mu^I+\mathcal{N}^A_a \Pi^a$, where the linear combination coefficients $\{\mu^I, \Pi^a \}$ are hydrodynamic fields of their own right (which are smooth if and only if $\varphi^A$ is smooth). Since the expansion of a vector on a basis is necessarily unique, the change of variables $\varphi^A \rightarrow \{\mu^I, \Pi^a \}$ is one-to-one. Plugging $\varphi^A=\mathcal{N}^A_I \mu^I+\mathcal{N}^A_a \Pi^a$ into the constitutive relations
\begin{equation}
     E^\mu= \dfrac{1}{2} E^\mu_{AB} \varphi^A \varphi^B\,, \, \qquad  j_I^\mu = \mathcal{N}^A_I E^\mu_{AB} \varphi^B \, , \qquad 
 \sigma =  \sigma_{AB} \varphi^A \varphi^B \,, 
\end{equation}
we obtain equation \eqref{decompose}, where the symmetry of $E^\mu_{IJ}$, $E^\mu_{ab}$, and $\sigma_{ab}$ follows from the symmetry of $E^\mu_{AB}$ and $\sigma_{AB}$. Finally, suppose that there is a vector $V^a \neq 0$ such that $V^a\sigma_{ab}=0$. Then, one can show that $V^a\mathcal{N}^A_a\sigma_{AB}=0$ (use $\sigma_{ab}=\sigma_{AB}\mathcal{N}^A_a\mathcal{N}^B_b$, and $\sigma_{AB}\mathcal{N}^B_I=0$), i.e. the matrix $\mathcal{N}^A_I$ is ``missing'' one conservation law, which can be accounted for by extending $\mathcal{N}^A_I$ to include $V^a\mathcal{N}^A_a$ as an additional (linearly independent) column. 
\end{proof}

\newpage

\section{Atlas of Information Currents and entropy production rates}

Below, we provide the information currents and entropy production rates of the theories listed in the table present in the main text.

\noindent \textit{How to interpret this atlas -} Both $E^\mu$ and $\sigma$ have been rescaled by a constant factor $T$ (background temperature) for convenience. The fields $\varphi^A$ describe small displacements from equilibrium, which are marked with a variation symbol ``$\delta$'', e.g. $\delta T$, or $\delta u^k$.  The coefficients in front of them (i.e., all the quantities without ``$\delta$'') are background quantities, which must be treated as constants. In some information currents, we have introduced a ``chemical'' index $X,Y \in \{1,...,a\}$. Einstein's summation convention applies to all types of indices, including $X$ and $Y$. Chemical matrices, such as $P_{XY}$ and $\mathcal{K}^{XY}$, are all symmetric.

\noindent \textit{How to use this atlas -} With the aid of Theorem 1, one can fully reconstruct the linear field equations of each theory directly from the formulas of $TE^\mu$ and $T\sigma$. In fact, the system $E^\mu_{AB}\partial_\mu \varphi^A = -\sigma_{AB}\varphi^B$ is formally equivalent to \cite{GavassinoNonHydro2022}
\begin{equation}\label{systemreads}
    \partial_\mu \dfrac{\partial (TE^\mu)}{\partial \varphi^A} = -\dfrac{1}{2} \dfrac{\partial (T\sigma)}{\partial \varphi^A}.
\end{equation}
For example, take the diffusion of a non-conserved density, class $(1,1,0)-(0,0,0)$. Plugging \eqref{diffusuz} into \eqref{systemreads}, we obtain
\begin{equation}\label{systemreads33}
\begin{split}
    \partial_\mu \dfrac{\partial (TE^\mu)}{\partial (\delta \mu)} ={}& -\dfrac{1}{2} \dfrac{\partial (T\sigma)}{\partial (\delta \mu)} \spc \Longrightarrow \spc \partial_t (A \delta \mu) + \partial_j \delta j^j =-\Xi \delta \mu \\
    \partial_\mu \dfrac{\partial (TE^\mu)}{\partial (\delta j^k)} ={}& -\dfrac{1}{2} \dfrac{\partial (T\sigma)}{\partial (\delta j^k)} \spc \Longrightarrow \spc \partial_t (\mathcal{K} \delta j_k) + \partial_k \delta \mu =-\mathcal{R} \delta j_k\\
    \end{split}
\end{equation}
As one can see, the system \eqref{systemreads33} is symmetric by construction. Once the field equations have been obtained, it is possible to assign an interpretation to the fields and the transport coefficients based on their role in the dynamics. For example, in the system \eqref{systemreads33}, we can interpret $A\delta \mu$ as the density of a non-conserved charge, $\delta \mu$ as its chemical potential (which should vanish at equilibrium), $\delta j^k$ as its flux, and $\Xi$ as its production rate. The second equation of \eqref{systemreads33} is a Cattaneo-type relaxation equation for the flux $\delta j^k$, with relaxation time $\mathcal{K}/\mathcal{R}$ and diffusivity coefficient $\mathcal{R}^{-1}$.

\subsection*{Chemistry: ($a,0,0$)-($\leq a, 0,0$)}

\begin{equation}
\begin{split}
TE^0 ={}& \dfrac{1}{2} P_{XY} \, \delta \mu^X \delta \mu^Y \, ,  \\
TE^j ={}& 0 \, , \\
T\sigma ={}& \Xi_{XY} \, \delta \mu^X \delta \mu^Y \, , \\
\end{split}
\end{equation}

\subsection*{Fluid mixtures: ($a,1,0$)-($\leq a, 1,0$)}

\begin{equation}
\begin{split}
TE^0 ={}& \dfrac{1}{2} P_{XY} \, \delta \mu^X \delta \mu^Y + \dfrac{1}{2}(\rho {+} P)\delta u_k \delta u^k \, ,\\
TE^j ={}& n_X \delta \mu^X \delta u^j \, , \\
T\sigma ={}& \Xi_{XY} \, \delta \mu^X \delta \mu^Y \, ,\\
\end{split}
\end{equation}

\subsection*{Carter multifluids: ($a,a,0$)-($\leq a, \leq a,0$)}

\begin{equation}
\begin{split}
TE^0 ={}& \dfrac{1}{2} P_{XY} \, \delta \mu^X \delta \mu^Y + \dfrac{1}{2}\mathcal{K}^{XY}\delta j_{Xk} \delta  j_Y^k \\
TE^j ={}&  \delta \mu^X \delta j_X^j \\
T\sigma ={}& \Xi_{XY} \, \delta \mu^X \delta \mu^Y + \mathcal{R}^{XY} \delta j_{Xk} \delta  j_Y^k\\
\end{split}
\end{equation}

\subsection*{Diffusion of a non-conserved density: ($1,1,0$)-($0, 0,0$)}

\begin{equation}\label{diffusuz}
\begin{split}
TE^0 ={}& \dfrac{1}{2} A \, (\delta \mu)^2 + \dfrac{1}{2}\mathcal{K}\, \delta j_{k} \delta  j^k \\
TE^j ={}&  \delta \mu \, \delta j^j \\
T\sigma ={}& \Xi \, (\delta \mu)^2 + \mathcal{R} \, \delta j_{k} \delta  j^k\\
\end{split}
\end{equation}

\subsection*{Cattaneo model of diffusion: ($1,1,0$)-($1, 0,0$)}

\begin{equation}
\begin{split}
TE^0 ={}& \dfrac{1}{2} A \, (\delta \mu)^2 + \dfrac{1}{2}\mathcal{K}\, \delta j_{k} \delta  j^k \\
TE^j ={}&  \delta \mu \, \delta j^j \\
T\sigma ={}&  \mathcal{R} \, \delta j_{k} \delta  j^k\\
\end{split}
\end{equation}

\subsection*{Barotropic perfect fluid: ($1,1,0$)-($1, 1,0$)}

\begin{equation}
\begin{split}
TE^0 ={}& \dfrac{1}{2} \dfrac{n^2 (\delta \mu)^2 }{(\rho {+} P)c_s^2}+ \dfrac{1}{2} (\rho {+} P) \, \delta u_{k} \delta  u^k \\
TE^j ={}&  n\delta \mu \, \delta u^j \\
T\sigma ={}& 0\\
\end{split}
\end{equation}

\subsection*{Bulk viscous fluid at $\mu=0$: ($2,1,0$)-($1, 1,0$)}

\begin{equation}
\begin{split}
TE^0 ={}& \dfrac{1}{2} \bigg[ \dfrac{c_v}{T} (\delta T)^2 + b_0 (\delta \Pi)^2 + (\rho {+} P)\, \delta u_{k} \delta  u^k \bigg] \\
TE^j ={}&  (s \delta T +\delta \Pi) \delta u^j \\
T\sigma ={}& \dfrac{(\delta \Pi)^2}{\zeta} \\
\end{split}
\end{equation}

\subsection*{Perfect fluid: ($2,1,0$)-($2, 1,0$)}

\begin{equation}
\begin{split}
TE^0 ={}& \dfrac{1}{2} \bigg[ P_{TT} (\delta T)^2 +2 P_{T\mu} \delta T \delta \mu + P_{\mu \mu} (\delta \mu)^2 + (\rho {+} P)\, \delta u_{k} \delta  u^k \bigg] \\
TE^j ={}&  (s \delta T +n \delta \mu) \delta u^j \\
T\sigma ={}& 0 \\
\end{split}
\end{equation}

\subsection*{Coupled diffusion of two conserved densities: ($2,2,0$)-($2, 0,0$)}

\begin{equation}
\begin{split}
TE^0 ={}& \dfrac{1}{2} \bigg[ P_{TT} (\delta T)^2 +2 P_{T\mu} \delta T \delta \mu + P_{\mu \mu} (\delta \mu)^2 + \mathcal{K}^{ss} \delta s^k \delta s_{k} + 2\mathcal{K}^{sn} \delta s^k \delta n_{k} + \mathcal{K}^{nn} \delta n^k \delta n_{k} \bigg] \\
TE^j ={}&  \delta T \, \delta s^j + \delta \mu \, \delta n^j  \\
T\sigma ={}& \mathcal{R}_{ss} \delta s^k \delta s_{k} + 2\mathcal{R}_{sn} \delta s^k \delta n_{k} + \mathcal{R}_{nn} \delta n^k \delta n_{k} \\
\end{split}
\end{equation}

\subsection*{Heat conductive bulk viscous fluid at $\mu=0$: ($2,2,0$)-($1, 1,0$)}

\begin{equation}
\begin{split}
TE^0 ={}& \dfrac{1}{2} \bigg[ \dfrac{c_v}{T} (\delta T)^2 + b_0 (\delta \Pi)^2 + (\rho {+} P)\, \delta u_{k} \delta  u^k + 2\delta u_k \delta q^k + b_1 \delta q_k \delta q^k \bigg]  \\
TE^j ={}&  (s \delta T +\delta \Pi) \delta u^j + (T^{-1}\delta T -a_0 \delta \Pi) \delta q^j \\
T\sigma ={}& \dfrac{(\delta \Pi)^2}{\zeta} + \dfrac{\delta q_k \delta q^k}{\kappa T} \\
\end{split}
\end{equation}

\subsection*{Heat conductive fluid at $\mu \neq 0$: ($2,2,0$)-($2, 1,0$)}

\begin{equation}
\begin{split}
TE^0 ={}& \dfrac{1}{2} \bigg[ P_{TT} (\delta T)^2 +2 P_{T\mu} \delta T \delta \mu + P_{\mu \mu} (\delta \mu)^2 + (\rho {+} P)\, \delta u_{k} \delta  u^k + 2\delta u_k \delta q^k + b_1 \delta q_k \delta q^k \bigg]  \\
TE^j ={}&  (s \delta T +n \delta \mu) \delta u^j + T^{-1}\delta T \delta q^j \\
T\sigma ={}&  \dfrac{\delta q^k \delta q_k}{\kappa T} \\
\end{split}
\end{equation}

\subsection*{Relativistic superfluid: ($2,2,0$)-($2, 2,0$)}

\begin{equation}
\begin{split}
TE^0 ={}& \dfrac{1}{2} \bigg[ P_{TT} (\delta T)^2 +2 P_{T\mu} \delta T \delta \mu + P_{\mu \mu} (\delta \mu)^2 + \mathcal{K}^{ss} \delta s^k \delta s_{k} + 2\mathcal{K}^{sn} \delta s^k \delta n_{k} + \mathcal{K}^{nn} \delta n^k \delta n_{k} \bigg] \\
TE^j ={}&  \delta T \, \delta s^j + \delta \mu \, \delta n^j  \\
T\sigma ={}&0\\
\end{split}
\end{equation}

\subsection*{Maxwell material at $\mu=0$: ($1,1,1$)-($1, 1,0$)}

\begin{equation}
\begin{split}
TE^0 ={}& \dfrac{1}{2} \bigg[  \dfrac{c_v}{T} (\delta T)^2+ (\rho {+} P)\, \delta u_{k} \delta  u^k + b_2 \delta \Pi^{jk} \delta \Pi_{jk} \bigg] \\
TE^j ={}& s \delta T \delta u^j + \delta \Pi^{jk} \delta u_k \\
T\sigma ={}& \dfrac{\delta \Pi^{jk} \delta \Pi_{jk}}{2\eta} \\
\end{split}
\end{equation}

\subsection*{Elastic material at $\mu=0$: ($1,1,1$)-($1, 1,1$)}

\begin{equation}
\begin{split}
TE^0 ={}& \dfrac{1}{2} \bigg[  \dfrac{c_v}{T} (\delta T)^2+ (\rho {+} P)\, \delta u_{k} \delta  u^k + \dfrac{ \delta \Pi^{jk} \delta \Pi_{jk}}{2G} \bigg] \\
TE^j ={}& s \delta T \delta u^j + \delta \Pi^{jk} \delta u_k \\
T\sigma ={}& 0 \\
\end{split}
\end{equation}

\subsection*{Burgers material at $\mu=0$: ($1,1,2$)-($1, 1,0$)}

\begin{equation}
\begin{split}
TE^0 ={}& \dfrac{1}{2} \bigg[  \dfrac{c_v}{T} (\delta T)^2+ (\rho {+} P)\, \delta u_{k} \delta  u^k + b_2 \delta \Pi^{jk} \delta \Pi_{jk}+b_2 \delta \Lambda^{jk} \delta \Lambda_{jk}\bigg] \\
TE^j ={}& s \delta T \delta u^j + \delta \Pi^{jk} \delta u_k \\
T\sigma ={}& \xi_1 \delta \Pi^{kl}\delta \Pi_{kl}+2\xi_2 \delta \Pi^{kl}\delta \Lambda_{kl}+\xi_3\delta \Lambda^{kl}\delta \Lambda_{kl} \\
\end{split}
\end{equation}

\subsection*{Israel-Stewart theory in a general frame at $\mu =0$: ($3,2,1$)-($1, 1,0$)}

\begin{equation}
\begin{split}
TE^0 ={}& \dfrac{1}{2} \bigg[  \dfrac{c_v}{T} (\delta T)^2+ 2\dfrac{\delta \mathcal{A} \delta T}{T}+b_{1} (\delta \mathcal{A})^2+2b_2 \delta \Pi \delta \mathcal{A}+ b_3 (\delta \Pi)^2+(\rho {+} P)\, \delta u^{k} \delta  u_k + 2 \delta u^k \delta q_k+ b_4 \delta q^k \delta q_k + b_5 \delta \Pi^{jk} \delta \Pi_{jk} \bigg] \\
TE^j ={}& (s \delta T +\delta \Pi)\delta u^j  +(T^{-1} \delta T -a_{1} \delta \mathcal{A}- a_2 \delta \Pi )\delta q^j + (\delta u_k -a_3 \delta q_k)\delta \Pi^{jk}  \\
T\sigma ={}& \xi_1 (\delta \mathcal{A})^2 + 2\xi_2 \, \delta \mathcal{A} \delta \Pi +\xi_3 (\delta \Pi)^2 + \dfrac{\delta q^k \delta q_k}{\kappa T} +\dfrac{\delta \Pi^{jk} \delta \Pi_{jk}}{2\eta} \\
\end{split}
\end{equation}

\subsection*{Israel-Stewart theory at $\mu \neq 0$: ($3,2,1$)-($2, 1,0$)}

\begin{equation}
\begin{split}
TE^0={}& \dfrac{1}{2} \bigg[P_{TT} (\delta T)^2 +2 P_{T\mu} \delta T \delta \mu + P_{\mu \mu} (\delta \mu)^2  + (\rho{+}P) \delta u^k \delta u_k + 2 \delta u^k \delta q_k + b_1 \delta q^k \delta q_k +b_0 (\delta \Pi)^2 + b_2 \delta \Pi^{jk} \delta \Pi_{jk} \bigg] \\
TE^j={}& (s\delta T+ n\delta \mu +\delta \Pi)\delta u^j +(T^{-1} \delta T -a_0 \delta \Pi)\delta q^j + ( \delta u_k -a_1 \delta q_k)\delta \Pi^{jk}  \\
T\sigma ={}&  \dfrac{(\delta \Pi)^2}{\zeta} + \dfrac{\delta q^k \delta q_k}{\kappa T} +\dfrac{\delta \Pi^{jk} \delta \Pi_{jk}}{2\eta}  \\
\end{split}
\end{equation}

\subsection*{Elastic heat conducting material: ($3,2,1$)-($3, 1,1$)}
\begin{equation}
\begin{split}
TE^0={}& \dfrac{1}{2} \bigg[P_{TT} (\delta T)^2 +2 P_{T\mu} \delta T \delta \mu + P_{\mu \mu} (\delta \mu)^2   + (\rho{+}P) \delta u^k \delta u_k + 2 \delta u^k \delta q_k + b_1 \delta q^k \delta q_k + \dfrac{(\delta \Pi)^2}{K}+ \dfrac{\delta \Pi^{jk} \delta \Pi_{jk}}{2G} \bigg]   \\
TE^j={}& (s\delta T +n\delta \mu+ \delta \Pi)\delta u^j +(T^{-1} \delta T -a_0 \delta \Pi)\delta q^j + ( \delta u_k -a_1 \delta q_k)\delta \Pi^{jk}   \\
T\sigma ={}&  \dfrac{\delta q^k \delta q_k}{\kappa T}  \\
\end{split}
\end{equation}

\subsection*{Isotropic supersolid\footnote{Note that in the Andreev-Lifshitz theory for supersolid hydrodynamics \cite{Andreev1969,Sears2010} the coefficients $\mathcal{Z}_1$ and $\mathcal{Z}_2$ are set to zero. Here, to keep the class ``universal'', we have included them, because they are allowed by symmetry.}: ($3,2,1$)-($3, 2,1$)}

\begin{equation}\label{infosupsol}
\begin{split}
TE^0 ={}& \dfrac{1}{2} \bigg[P_{TT} (\delta T)^2 +2 P_{T\mu} \delta T \delta \mu + P_{\mu \mu} (\delta \mu)^2   + \mathcal{K}^{ss}\delta s^k \delta s_k + 2\mathcal{K}^{sn}\delta s^k \delta n_k + \mathcal{K}^{nn}\delta n^k \delta n_k + \dfrac{(\delta \Pi)^2}{K} {+} \dfrac{\delta \Pi^{kl}\delta \Pi_{kl}}{2G} \bigg]   \\
TE^j ={}& \delta T \delta s^j +\delta \mu \delta n^j + (s^{-1} \delta s^j-\mathcal{Z}_1 \delta n^j) \delta \Pi  + (s^{-1} \delta s_k -\mathcal{Z}_2 \delta n_k) \delta \Pi^{jk}    \\
T\sigma ={}& 0  \\
\end{split}
\end{equation}

\label{lastpage}

\end{document}